\documentclass[11pt, letter]{article}
\usepackage{amsmath}
\usepackage{amssymb}
\usepackage{amsfonts}
\usepackage{array}
\usepackage{float}
\usepackage{enumerate}
\usepackage{fullpage}
\usepackage{ifthen}
\usepackage{xspace}

\usepackage{graphicx}
\usepackage{arydshln}

\usepackage{algorithmic}
\usepackage{algorithm}

\floatstyle{ruled}
\newfloat{algorithm}{thp}{lop}[section]
\floatname{algorithm}{Algorithm}
\newfloat{module}{thp}{lop}[section]
\floatname{module}{Module}

\newtheorem{thm}{Theorem}[section]

\newtheorem{lem}[thm]{Lemma}
\newtheorem{definition}{Definition}[section]
\newtheorem{note}{Note}[section]
\newcommand{\qedsymb}{\hfill{\rule{2mm}{2mm}}}
\newenvironment{proof}
{\begin{trivlist}
\item[\hspace{\labelsep}{\bf\noindent Proof: }]
}
{\qedsymb\end{trivlist}}

\newcommand{\remove}[1]{}

\newenvironment{th-repeat}[1]{\begin{trivlist}
\item[\hspace{\labelsep}{\bf\noindent Theorem~\ref{#1} }]}%
{\end{trivlist}}

\newenvironment{lem-repeat}[1]{\begin{trivlist}
\item[\hspace{\labelsep}{\bf\noindent Lemma~\ref{#1} }]}%
{\end{trivlist}}

\begin{document}

\title{Optimal Byzantine-resilient Convergence in Oblivious Robot Networks}

\author{Zohir Bouzid \and Maria Gradinariu Potop-Butucaru \and
S\'{e}bastien Tixeuil}

\date{Universit\'e Pierre et Marie Curie - Paris 6, LIP6-CNRS 7606, France}

\maketitle

\begin{abstract}
Given a set of robots with arbitrary initial location and no agreement on a global coordinate system, \emph{convergence} requires that all robots asymptotically approach the exact same, but unknown beforehand, location. Robots are oblivious--- they do not recall the past computations --- and are allowed to move in a one-dimensional space. Additionally, robots cannot communicate directly, instead they obtain system related information only \emph{via} visual sensors. 

We draw a connection between the convergence problem in robot networks, and the distributed \emph{approximate agreement} problem (that requires correct processes to decide, for some constant $\epsilon$, values distance $\epsilon$ apart and within the range of initial proposed values). Surprisingly, even though specifications are similar, the convergence implementation in robot networks requires specific assumptions about synchrony and Byzantine resilience. 

In more details, we prove necessary and sufficient conditions for the convergence of mobile robots despite a subset of them being Byzantine (\emph{i.e.} they can exhibit arbitrary behavior). Additionally,  we propose a deterministic convergence algorithm for robot networks and analyze its correctness and complexity in various synchrony settings. The proposed algorithm tolerates $f$ Byzantine robots for $(2f+1)$-sized robot networks in fully synchronous networks, $(3f+1)$-sized in semi-synchronous networks. 
These bounds are optimal for the class of \emph{cautious} algorithms, which guarantee that correct robots always move inside the range of positions of the correct robots.
\end{abstract}

\section{Introduction} 

The execution of complex tasks in hostile environments (\emph{e.g.}
oceans or planets exploration,  decontamination of radioactive areas,
human search and rescue operations) makes necessary the use of robots as an alternative to human intervention. 
So far, robots have been studied mainly through the prism of engineering or artificial intelligence, with success in the case of single powerful robots. However, many of the envisioned new tasks can not or should not (for cost reasons)  be achieved by an unique robot, hence low cost swarm of cheap mobile robots executing coordinated tasks in a distributed manner is appealing when considering
dangerous environments.
The study of autonomous swarms of robots is also a challenging area for distributed computing, as networks of robots raise a variety of problems related to distributed control and coordination. 

In order to capture the difficulty of distributed coordination in
robot networks two main computational models are proposed in the
literature: the ATOM \cite{SY99} and CORDA models \cite{Pre05}. In
both models robots are considered identical and indistinguishable, can
see each other via visual sensors and operate in look-compute-move
cycles. Robots, when activated, observe the location of the other
robots in the system, compute a new location and move accordingly. The
main difference between the two models comes from the granularity of
the execution of this cycle. 
In the ATOM model, robots executing concurrently
are in phase while in CORDA they are asynchronous (i.e. a robot can
execute the look action for example while another robot performs its move action).

Gathering and convergence are two related fundamental tasks in robot
networks. Gathering requires  robots to \emph{reach} a single point
within finite time regardless of their initial positions while
convergence only requires robots to get \emph{close} to a single
point. More specifically, $\forall \epsilon>0$, there is a time 
$t_\epsilon$ such that all robots are at distance at most $\epsilon$
from each other. 
Gathering and convergence can serve as the basis of many other
protocols, such as constructing a 
common coordinate system or arranging themselves in a specific geometrical pattern.

\paragraph{Related works}

Since the pioneering work of Suzuki and Yamashita~\cite{SY99},
gathering and convergence have been addressed  in \emph{fault-free}
systems for a broad class of settings. Prencipe~\cite{Pre05} studied the problem of gathering in both ATOM and CORDA models, and showed that the problem is intractable without additional assumptions such as being able to detect the multiplicity of a location (\emph{i.e.}, knowing the number of robots that may simultaneously occupy that location). Flocchini \emph{et~al.}~\cite{FPS+05} proposed a gathering solution for oblivious robots with limited visibility in CORDA model, where robots share the knowledge of a common direction given by a compass. The subsequent work by Souissi \emph{et~al.}~\cite{TR:SDY05} consider a system in which compasses are not necessarily consistent initially. Ando \emph{et~al.}~\cite{ando1999dmp} propose a gathering algorithm for the ATOM model with limited visibility. 

The case of \emph{fault-prone} robot networks was recently tackled by several academic studies. Cohen and Peleg~\cite{CP06} study the problem when robots observations and movements are subject to errors. Fault tolerant gathering is addressed in \cite{agmon2004ftg} where the authors study a gathering protocol that tolerates one crash (\emph{i.e.} one robot may stop moving forever), and they also provide an algorithm for the ATOM model with fully synchronous scheduling that tolerates up to $f$ byzantine faults (\emph{i.e.} $f$ robots may exhibit arbitrary behavior), when the number of robots is (strictly) greater than $3f$. In \cite{defago3274fta} the authors study the feasibility of gathering in crash-prone and Byzantine-prone environments and propose probabilistic solutions altogether with detailed analysis relating scheduling and problem solvability.

The specification of convergence being less stringent than that of gathering, it is worth investigating whether this leads to better fault and Byzantine tolerance. In~\cite{ando1999dmp} the authors address convergence with limited visibility in fault-free environments. Convergence with inaccurate sensors and movements is addressed in \cite{cohen6cam}. Fault-tolerant convergence was first addressed in~\cite{cohen2004rcv,cohen2005cpg}, where algorithms based on the convergence to the center of gravity of the system are presented. Those algorithms work in CORDA model and tolerate up to $f$ ($n>f$) crash faults, where $n$ is the number of robots in the system. To our knowledge, none of the aforementioned works on convergence addresses the case of byzantine faults. 

\paragraph{Our contributions}

In this paper we focus on the feasibility of deterministic solutions
for convergence in robots networks that are prone to byzantine
faults and move in a uni-dimentional space. Our contribution is threefold:
\begin{enumerate}
\item We draw a connection between the convergence problem in robot
networks, and the distributed \emph{approximate agreement} problem
(that requires correct processes to decide, for some constant
$\epsilon$), values distance $\epsilon$ apart and within the range of initial values. In particular, our work uses a similar technique as the one presented in \cite{dolev1986raa} and \cite{abraham2005ora} for the problem of approximate agreement with byzantine failures. They propose approximate agreement algorithms that tolerate up to $f$ byzantine failures and require $n>3f$, which has been proven optimal for both the synchronous and asynchronous case.
\item We prove necessary and sufficient conditions for the convergence of mobile robots despite a subset of them being Byzantine (\emph{i.e.} that can exhibit arbitrary behavior), when those robots can move in a uni-dimensional space.
\item We propose a deterministic convergence algorithm for robot networks and analyze its correctness and complexity in various synchrony settings.
The proposed algorithm tolerates $f$ Byzantine robots for $(2f+1)$-sized robot networks in fully synchronous networks, $(3f+1)$-sized in semi-synchronous networks. 
These bounds are optimal for the class of \emph{cautious} algorithms, which guarantee that correct robots always move inside the range of positions of other correct robots.
\end{enumerate}

\paragraph{Outline}
The remaining of the paper is organized as follows: Section~\ref{sec:model} presents our model and robot network assumptions, Sections~\ref{sec:cns} and \ref{sec:lower-bounds} provide necessary and sufficient conditions for the convergence problem with Byzantine failures, Section~\ref{sec:algorithm} describes our protocol and its complexity, while concluding remarks are presented in Section~\ref{sec:conclusion}.

\section{Preliminaries}
\label{sec:model}

Most of the notions presented in this section are borrowed from \cite{SY99,Pre01,agmon2004ftg}.  We consider a network that consists of a finite set of robots arbitrarily deployed in a uni-dimensional space. The robots are devices with sensing, computing and moving capabilities. They can observe (sense) the positions of other robots in the space and based on these observations, they perform some local computations that can drive them to other locations. 

In the context of this paper, the robots are \emph{anonymous}, in the sense that they can not be distinguished using their appearance, and they do not have any kind of identifiers that can be used during the computation. In addition, there is no direct mean of communication between them. Hence, the only way for robots to acquire information is by observing each others positions. Robots have \emph{unlimited visibility}, \emph{i.e.} they are able to sense the whole set of robots.
Robots are also equipped with a multiplicity sensor. This sensor is referred as \emph{simple multiplicity detector}, denoted by $\mathcal{M?}$, if it can distinguish if there are more than one robot at a given position. If it can also detect the exact number of robots collocated in the same point, it is referred as \emph{multiples detector}, denoted in the sequel by $\mathcal{M}$. 
We prove in this paper that $\mathcal{M}$ is necessary in order to deterministically solve the convergence problem in a uni-dimensional space even in the presence of a single Byzantine robot. 

\subsection{System model}
A robot that exhibits discrete behavior is modeled with an I/O automaton~\cite{lynch1996da}, while one with continous behavior will be modeled using a hybrid I/O automaton~\cite{lynch2003ha}. The actions performed by the automaton that models a robot are as follows:
\begin{itemize}
\item \emph{Observation (input type action)}.
  An observation returns a snapshot of the positions of all robots within the visibility range. In our case, this observation returns a snapshot of the positions of \emph{all} robots denoted with $P(t)= \{P_1(t), ... , P_n(t)\}$. The positions of correct robots are referred as $U(t)=\{U_1(t), ... , U_m(t)\}$ such that $m \geq n-f$. Note that $U(t) \subseteq P(t)$. The observed positions are \emph{relative} to the observing robot, that is, they use the coordinate system of the observing robot.
\item \emph{Local computation (internal action)}.
  The aim of this action is the computation of a destination point (possibly using the relative position of other robots that was previously observed);
  \item \emph{Motion (output type action)}.
  This action commands the motion of robots towards the destination location computed in the previous local computation action.
\end{itemize}

The ATOM or SYm model addressed in this paper considers discrete time at irregular intervals. At each
time, some subset of the robots become active and complete an entire
computation cycle composed of the previously described elementary actions (observation, local computation and motion). 
Robots can be active either simultaneously or
sequentially.  Two robots that are active simultaneously observe the
exact same environment (according to their respective coordinate
systems).
 
The \emph{local state} of a robot at time~$t$ is the state of its input/output variables and the state of its local variables and registers. A \emph{network of robots} is modeled by the parallel composition of the individual automata that model each robot in the network. A \emph{configuration} of the system at time~$t$ is the union of the local states of the robots in the system at time~$t$. An \emph{execution} $e=(c_0, \ldots, c_t, \ldots)$ of the system is an infinite sequence of configurations, where $c_0$ is the initial configuration\footnote{Unless stated otherwise, we make no specific assumption regarding the respective positions of robots in initial configurations.} of the system, and every transition $c_i \rightarrow c_{i+1}$ is associated to the execution of a cycle by a subset of robots.

A \emph{scheduler} can be seen as an entity that is external to the system and selects robots for execution. As more power is given to the scheduler for robot scheduling, more different executions are possible and more difficult it is to design robot algorithms.
In the remaining of the paper, we consider that the scheduler is \emph{fair}, that is, in any infinite execution, every robot is activated infinitely often. A scheduler is \emph{$k$-bounded} if, between any two activations of a particular robot, any other robot can be activated at most $k$ times. The particular case of the \emph{fully synchronous} scheduler activates all robots in every configuration. Of course, an impossibility result for a more constrained scheduler (\emph{e.g.} bounded) also holds for a less constrained one (\emph{e.g.} fair), and an algorithm for the fair scheduler is also correct in for the $k$-bounded scheduler or the fully-synchronous scheduler. The converse is not necessarily true. 
 
The faults we address in this paper are \emph{Byzantine} faults. A byzantine (or malicious) robot may behave in arbitrary and unforeseeable way. In each cycle, the scheduler determines the course of action of faulty robots and the distance to which each non-faulty robot will move in this cycle. However, a robot $i$ is guaranteed to move a distance of at least $\delta_i$ towards its destination before it can be stopped by the scheduler.

Our convergence algorithm performs operations on multisets. A multiset or a bag $S$ is a generalization of a set where an element can have more than one occurence. The number of occurences of an element \emph{a} is referred as its \emph{multiplicity} and is denoted by $mul(a)$. The total number of elements of a multiset, including their repeated occurences, is referred as the \emph{cardinality} and is denoted by $|S|$. $\min(S)$(resp. $\max(S)$) is the smallest (resp. largest) element of $S$. If $S$ is nonempty, $range(S)$  denotes the set $[\min(S), \max(S)]$ and $diam(S)$ (diameter of $S$) denotes $\max(S) - \min(S)$.

\subsection{The Byzantine Convergence Problem}

In the following we refine the definition of the \emph{point convergence problem} from~\cite{agmon2004ftg}: given an initial configuration of $N$ autonomous mobile robots $M$ of which are correct ($M \geq N-f$), for every $\epsilon > 0$, there is a time $t_\epsilon$ from which all correct robots are within distance of at most $\epsilon$ of each other.

\begin{definition}[Byzantine Convergence]
\label{def:byz-convergence}
A system of oblivious robots verify the Byzantine convergence specification if and only if $\forall \epsilon > 0, \exists t_\epsilon$ such that $\forall t > t_\epsilon$, $\forall$ i,j $\leq M, \mathit{distance}(U_i(t), U_j(t)) < \epsilon$, where $U_i(t)$ and $U_j(t)$ are the positions of some \emph{correct} robots $i$ and $j$ at time $t$, and where $\mathit{distance}(a,b)$ denote the Euclidian distance between two positions. 
\end{definition}

Definition~\ref{def:byz-convergence} requires the convergence property only from the \emph{correct} robots. Note that it is impossible to obtain the converge of all the robots in the system regardless their behavior since Byzantine robots may exhibit arbitrary behavior and never join the position of correct robots.

\section{Necessary and sufficient conditions for deterministic convergence}
\label{sec:cns}

In this section we address the necessary and sufficient conditions to achieve convergence of robots in systems prone to byzantine failures. We define \emph{shrinking} algorithms (algorithms that eventually decrease the range among correct robots) and prove that this condition is necessary but not sufficient for convergence even in fault-free environments. We then define \emph{cautious} algorithms (algorithms that ensure that the position of correct robots always remains inside the range of the correct robots) and show that this condition, combined with the previous one, is sufficient to reach convergence in fault-free systems. Moreover, we address the necessary and sufficient conditions for convergence in byzantine-prone environments and show that for the problem to admit solutions additional assumptions (\emph{e.g.} multiplicity knowledge) are necessary.

\subsection{Necessary and sufficient conditions in fault-free environments}

By definition, convergence aims at asymptotically decreasing the range of possible values for the correct robots. The shrinking property captures this property. An algorithm is shrinking if there exists a constant factor $\alpha \in (0,1)$ such that starting in any configuration the range of correct robots eventually decreases by a multiplicative $\alpha$ factor. 

\begin{definition}[Shrinking Algorithm] 
An algorithm is \emph{shrinking} if and only if $\exists \alpha \in (0,1)$ such that $\forall t$, if $diam(U(t)) \neq 0, \exists t^\prime > t$, such that $diam(U(t^\prime)) < \alpha*diam(U(t))$, where $U(t)$ is the multiset of positions of correct robots.
\end{definition}

\begin{figure}[htbp]
\begin{center}
\includegraphics[height=6cm]{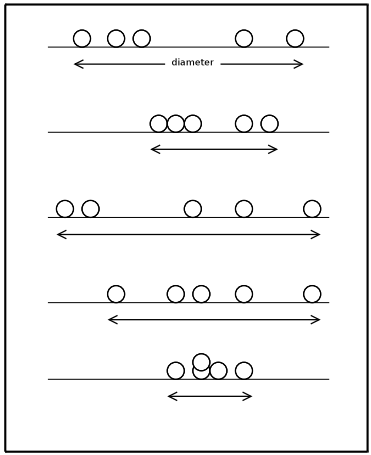}
\caption{Oscillatory effect of a shrinking algorithm}
\label{fig-oscillatory}
\end{center}
\end{figure}

\begin{note}
Note that the definition does not imply that the diameter always remains smaller than $\alpha*diam(U(t))$ after $t^\prime$ (see Figure~\ref{fig-oscillatory}). Therefore, an oscillatory effect is possible: the algorithm alternates between periods where the diameter is increased and decreased. However, each increasing period is followed by a decreasing one as depicted in Figure~\ref{fig-oscillatory}. Therefore a shrinking algorithm is not necessarily convergent.
\end{note}

\begin{lem}
\label{lem:convergence}
Any algorithm solving the convergence problem is necessarily shrinking.
\end{lem}

\begin{proof}
Assume that a convergence algorithm is not shrinking.
Then there exists some constant factor $\alpha \in (0,1)$, and some
time instant $t_1$ such that 
the diameter of correct robots after $t_1$ never decreases by a factor
of $\alpha$ i.e. 
$diam(U(t_2))$ is greater than $\alpha*diam(U(t_1))$ for any $t_2 > t_1$.
Therefore, there will always exist two correct robots that are at
distance of at least 
$\alpha*diam(U(t_1))$, which contradicts the assumption that the algorithm is convergent.
\end{proof}

A natural way to solve convergence is to never let the algorithm increase the diameter of correct robot positions. We say in this case that the algorithm is \emph{cautious}. 
A cautious algorithm is particularly appealing in the context of byzantine failures since it always instructs a correct robot to move inside the range of the positions held by the correct robots regardless of the locations of Byzantine ones. 
The notion of cautiousness was introduced \cite{dolev1986raa} in the context of classical Byzantine-tolerant distributed systems. In the following, we customize the definition of cautious algorithms for robot networks.

\begin{definition}[Cautious Algorithm]
Let $D_i(t)$ be the latest computed destination of robot $i$ up to time $t$ and let $U(t)$ be the positions of the correct robots at time $t$. \footnote{If the latest computation was executed at time $t^\prime \leq t$ then $D_i(t) = D_i(t^\prime)$.} An algorithm is \emph{cautious} if it satisfies the following two conditions:
\begin{itemize}
\item \textbf{cautiousness: } $\forall t,~D_i(t) \in range(U(t))$ for each robot $i$.
\item \textbf{non-triviality: } $\forall t$, if $\exists \epsilon > 0, \exists i,j < M, distance(U_i(t),U_j(t))\geq \epsilon$ (where $U_j(t)$ and $U_i(t)$ denote the positions of two correct robots $i$ and $j$ a time $t$), then $\exists t^\prime>t$ and a correct robot $k$ such that $D_k(t^\prime)\neq U_k(t^\prime)$
(at least one correct robot changes its position whenever convergence is not achieved).
\end{itemize}
\end{definition}

Note that the non-triviality condition ensures progress. That is, it prevents trivial solutions where each robot stays at its current position forever. 

The following two lemmas state some properties of cautious algorithms.

\begin{lem}
\label{diam-SS} 
In the ATOM model, if an algorithm is cautious then $\forall t^\prime >t~ diam(U(t^\prime)) \leq diam(U(t))$.
\end{lem}

\begin{proof}
Assume that it is not the case. i.e. that
$diam(U(t^\prime))>diam(U(t))$ for some $t^\prime>t$. Then there
exists two 
\textbf{successive} time instants, referred in the following cycles, $t_2>t_1$ such that $t\leq
t_1<t^\prime$,  $t<t_2 \leq t^\prime$ and 
the diameter of correct robots at $t_2$ is strictly greater than the diameter at $t_1$ i.e. $diam(U(t_2))>diam(U(t_1))$.
Thus, there exists at least one correct robot, say $r_1$, that was
inside $range(U(t_1))$ at $t_1$, 
and moved outside it at $t_2$. We prove that this is impossible.

Since cycles are atomic, no robot can move between $t_1$ and the LOOK
step of $t_2$, and the resulting 
snapshot of correct robots at this step is equal to $U(t_1)$. 
Thus, the destination point calculated by $r_1$ at $t_2$ is
necessarily inside $range(U(t_1))$ since the algorithm is cautious.

This contradicts the assumption that $r_1$ moves outside $range(U(t_1))$ at $t_2$, and the lemma follows.
\end{proof}
\begin{thm}
\label{lem:cands}
Any algorithm that is both cautious and shrinking solves the convergence problem in fault-free robot networks.
\end{thm}

\subsection{Necessary and sufficient conditions in Byzantine-prone environments}

In~\cite{Pre05}, Prencipe showed that multiplicity detection is necessary to achieve gathering without additional assumption. The situation is different when only convergence is requested (\emph{e.g.} the algorithm proposed in~\cite{cohen2005cpg} where no such condition is assumed). Interestingly enough, in the following we show that when robots are prone to Byzantine failures, a strongest kind of multiplicity detection becomes necessary in order to enable convergence \emph{via} cautious algorithms. Note that in the presence of byzantine faults, many multiplicity points (\emph{i.e.} points with multiple robots) may be created by the Byzantine robots. Moreover, if the trajectories of two robots intersect, it is relatively easy for the scheduler to stop those robots at exactly the same point to create an additional point of multiplicity. 

We show in the sequel that a simple multiplicity detector $\mathcal{M?}$ that can only distinguish whether multiple robots are at a given position (without returning the exact number of those robots) is not sufficient for cautious algorithms. A stronger detector, referred as \emph{multiples detector} $\mathcal{M}$, that can detect the exact number of robots collocated in the same point, is necessary. 

\begin{lem}
\label{lem:nec}
It is impossible to reach convergence with a cautious algorithm in Byzantine-prone environments with multiplicity detection, even in the presence of a single Byzantine fault.
\end{lem}

\begin{proof}
Let $A$ and $B$ be two distinct points in a uni-dimensional space (see
Figure \ref{fig-multiplicity-necessary}), 
and consider a set of robots without any multiplicity detection
capability. We suppose that it is possible to achieve convergence in
this case in presence of a 
single byzantine robot and we show that this leads to a contradiction.

\begin{figure}[htbp]
\begin{center}
\includegraphics[width=7cm]{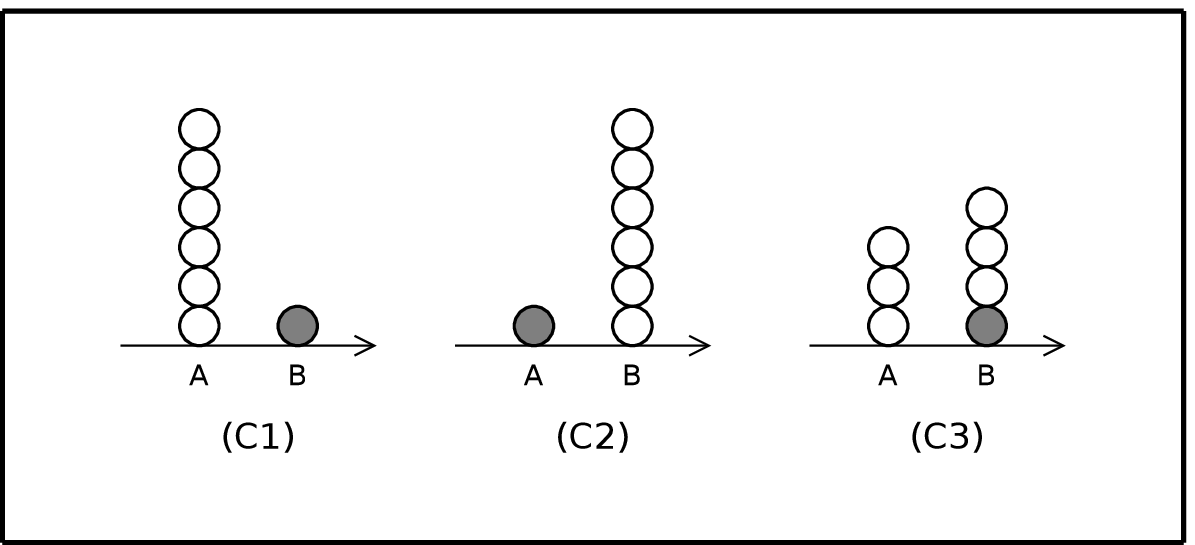}
\caption{Necessity of multiplicity detection to achieve convergence (black robots are byzantine)}
\label{fig-multiplicity-necessary}
\includegraphics[width=7cm]{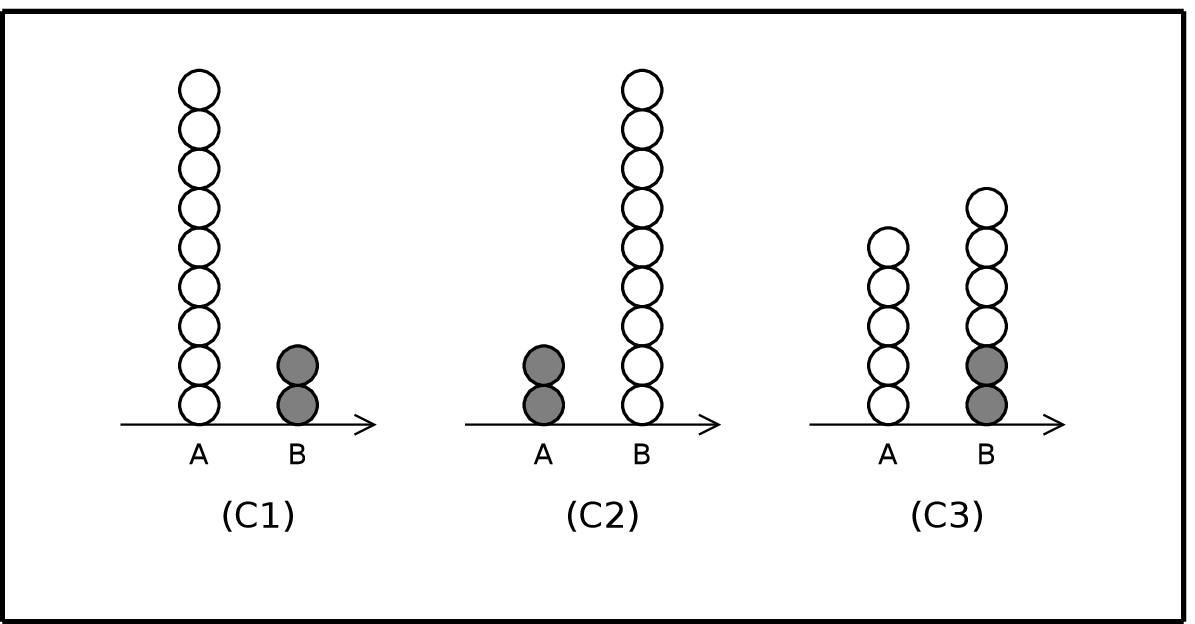}
\caption{Necessity of multiplicity number detection to achieve convergence (black robots are byzantine)}
\label{fig-multiplicity-nb-necessary}
\end{center}
\end{figure}

\begin{enumerate}
\item Let $C_1$ be a configuration where all correct robots are at $A$, and one byzantine robot at $B$.
If the robots at $A$ move, the scheduler can stop them at different
locations which causes the diameter of 
correct robots to increase which contradicts Lemma~\ref{diam-SS}, so they must stay at $A$.
\item Similarly, let $C_2$ be the symmetric configuration where the byzantine robot is at $A$, and the correct ones at $B$.  Then the robots at $B$ cannot move.
\item Let $C_3$ be a configuration where correct robots are spread over $A$ and $B$.  the byzantine robot may be indifferently at $A$ or at $B$. 
Since the robots are not endowed with a multiplicity detection
capability, the configurations $C_1$, $C_2$ and $C_3$ 
are indistinguishable to them. So they stay at their locations and the
algorithm never converge 
which contradicts the assumption that convergence is possible in this case.
\end{enumerate}
This proves that at least a simple multiplicity detector is necessary
to achieve convergence even if a single robot 
is byzantine.
\end{proof}

\begin{lem}
\label{lem:suff}
Multiples detection is necessary to reach Byzantine-tolerant convergence in a uni-dimensional space via cautious algorithms.
\end{lem}

\begin{proof}
The Algorithm \ref{alg:1} is a cautious algorithm that converges under
the assumption of multiples detection.
The previous lemma show that convergence cannot be achieved without
additional assumptions. Hence we consider the minimal set of
assumptions: robots are endowed with a simple multiplicity detection.
In the following we assume that convergence can be achieved with only 
simple multiplicity detection and we show that this leads to a contradiction.
Consider a set of robots in a uni-dimensional space prone to byzantine
failures and endowed 
with simple multiplicity detectors. The robots are spread over two
distinct points of the uni-dimensional 
space $A$ and $B$ (see figure \ref{fig-multiplicity-nb-necessary}).
\begin{enumerate}
\item Let $C_1$ be a configuration where all correct robots are at $A$,
and byzantine ones at $B$.  
We suppose that the number of correct robots at $A$ is sufficiently large to tolerate the byzantine robots of $B$. 
If the robots at $A$ move, they may be stopped by the scheduler at
different locations which increase their diameter. This contradicts
Lemma~\ref{diam-SS} because 
the algorithm is cautious. So the correct robots stay at $A$.
\item Consider the symmetric configuration $C_2$ where the correct
robots are at $B$ and 
the byzantine ones at $A$. With the same argument as $C_1$ we find that robots at $B$ stay there.
\item Let $C_3$ be a configuration where correct and byzantine robots are spread evenly between $A$ and $B$. 
Since robots are endowed only with simple multiplicity detectors, the
configurations $C_1$, $C_2$ and $C_3$ are indistinguishable to
them. So no 
robot will move and the algorithm never converges. This proves the lemma. 
\end{enumerate}
\end{proof}

\section{Lower bounds for byzantine resilient convergence}
\label{sec:lower-bounds}

In this section we study the lower bounds for Byzantine-resilient convergence of mobile robots in both fully and semi-synchronous ATOM models. The following lemma shows that any cautious algorithm needs at least $2f+1$ robots in order to tolerate $f$ byzantine robots.

\begin{lem}
\label{impossibility-2f}
It is impossible to achieve convergence with a cautious algorithm if $n\leq 2f$ in the fully-synchronous ATOM model, where $n$ denotes the number of robots and $f$ denotes the number of Byzantine robots.
\end{lem}

\begin{proof}
We assume that convergence is possible for $n \leq 2f$ and we show that this leads to a contradiction. We consider a set of $n$ robots, $f$ of which are faulty and assume the robots are spread over two points of the uni-dimensional space: $A$ and $B$. There are $f$ robots at point $A$ and $n-f$ robots at point $B$. Note that because $n \leq 2f$, each point contains at least $n-f$ robots (See figure \ref{fig-lower-bounds-SYM} for the case where $n=5$ and $f=3$).

Let $C_1$ be a configuration where all the correct robots $(n-f)$ are at $A$. The diameter is equal to $0$ and by Lemma\ref{diam-SS}, the diameter of correct robots never decreases if the algorithm is cautious. So the robots at~$A$ can not move. Otherwise, the diameter may increase.

Let $C_2$ be a configuration where all the correct robots are at $B$. These must not move and the argument is similar to the precedent case.

Let $C_3$ be a configuration where the correct robots are spread over
$A$ and $B$. Since the three configurations $C_1$, $C_2$ and $C_3$ 
are indistinguishable, the robots at $A$ and $B$ do not move and the
algorithm never converges, 
which contradicts the assumption that convergence is possible with $n \leq 2f$.

\begin{figure}[htbp]
\begin{center}
\includegraphics[width=7cm]{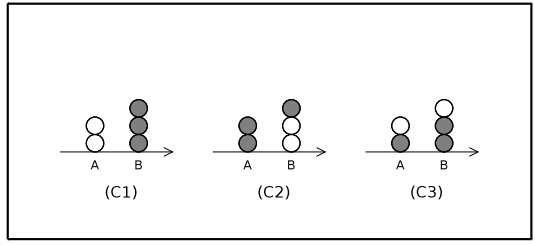}
\caption{Lower bounds for convergence in fully-synchronous ATOM model ($n$ = 5, $f$ = 3, black robots are byzantine)}
\label{fig-lower-bounds-SYM}
\end{center}
\end{figure}
\end{proof}

The following lemma provides the lower bound for the semi-synchronous case. 

\begin{lem}
\label{lem:lb}
Byzantine-resilient convergence is impossible for $n \leq 3f$ with a cautious algorithm in the semi-synchronous ATOM model and a $2$-bounded scheduler.
\end{lem}

\begin{proof}
By Lemma \ref{impossibility-2f}, convergence is impossible for $n \leq
2f$ in the fully-synchronous ATOM model, so it is also impossible in the semi-synchronous case.
Assume that there exists a cautious algorithm that achieves
convergence with $2f < n \leq 3f$.

Let A and B be two distinct points in a uni-dimensional space such that $(n-f)$ robots are located at A 
and the remaining f robots are located at B (see Figure \ref{fig-impossibility-3f} for $n=6$ and $f=2$).

\begin{figure}[htbp]
\begin{center}
\includegraphics[width=7cm]{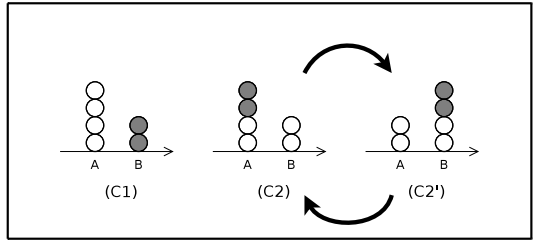}
\caption{Impossibility of convergence in SYM with $n\leq 3f$, black robots are byzantine}
\label{fig-impossibility-3f}
\end{center}
\end{figure}

Let $C_1$ be a configuration where all correct robots are at A and the byzantine ones at B. Note that since the correct robots are at the same point, the diameter is 0.
There are two possible cases:
\begin{enumerate}

\item The robots at A move when activated: since the algorithm is cautious, the only possible direction is towards B. 
When moving towards B, the robots may be stopped by the scheduler at different locations which causes the diameter to increase and this contradicts Lemma \ref{diam-SS}.

\item The robots at A do not move: the only possible action for robots in this configuration since they cannot move.
\end{enumerate}

Let $C_2$ be another possible configuration where the byzantine robots are at A, and the correct ones are spread over A and B as follows: $f$ correct robots at B and the remaining $(n-2f)$ at A.

 Note that $C_1$ and $C_2$ are indistinguishable by the individual robots
and assume the following scenario: The scheduler activates robots at
A. Since the configurations $C_1$ and $C_2$ are equivalent, robots
at A do not move. Then, the scheduler moves $n-2f \leq f$ faulty
robots from A to B which leads to the symmetric 
configuration $C_2^\prime$ and robots at B do not move neither. The
same scenario is repeated 
infinitely and no robot will ever move which prevents the algorithm to converge.
\end{proof}

\section{Deterministic Approximate Convergence}
\label{sec:algorithm}

In this section we propose a deterministic convergence algorithm and prove its correctness and optimality in the ATOM model. Algorithm~\ref{alg:1}, similarly to the approximate agreement algorithm in~\cite{dolev1986raa}, uses two functions, $\mathit{trim}_f(P(t))$ and $\mathit{median}(P(t))$. The former removes the $f$ largest and $f$ smallest values from the multiset given in parameter. The latter returns the median point in the input range. Using Algorithm~\ref{alg:1}, each robot computes the median of the positions of the robots seen in its last LOOK cycle ignoring the $f$ largest and $f$ smallest positions.

\begin{algorithm}
\caption{Byzantine Tolerant Convergence}          
\label{alg-FS}                  
\begin{algorithmic}
\STATE \textbf{Functions}:\\
\STATE $trim_f$: removes the $f$ largest and $f$ smallest values from the multiset given in parameter.
\STATE $median$: returns the points that is in the middle of the range of points given in parameter.
\STATE

\STATE   \textbf{Actions}:\\

\STATE move towards  $median(trim_f(P(t)))$
\end{algorithmic}
\label{alg:1}
\end{algorithm}

In the following we prove the correctness of Algorithm \ref{alg:1} in fully-synchronous and 
semi-synchronous ATOM models. 

In order to show that Algorithm \ref{alg:1} is convergent we prove first that it is cautious then we prove that it satisfies the specification of a shrinking algorithm. 

\subsection{Properties of Algorithm \ref{alg:1}}
In this section we propose a set of lemmas that will be further used in the construction of the convergence proof of our algorithm. 
In the following we recall a result related to the functions $trim$ and $range$ proved in \cite{dolev1986raa}.

\begin{lem}[\cite{dolev1986raa}]
\label{range-trim}
$range(trim_f(P(t))) \subset range(U(t))$.
\end{lem}

A direct consequence of the above property is that Algorithm
\ref{alg:1} is cautious for $n>2f$.
\begin{lem}
\label{cautious-FS}
Algorithm \ref{alg:1} is cautious for $n>2f$.
\end{lem}

\begin{lem}
\label{inside_range}
$range(trim_f(U(t))) \subseteq  range(trim_f(P(t)))$ when $n>3f$.
\end{lem}

\begin{proof}
We prove that: 
\begin{enumerate}
\item $\forall t~ U_{f+1}(t) \in range(trim_f(P(t))).$
\item $\forall t~ U_{m-f}(t) \in range(trim_f(P(t))).$
\end{enumerate}

\begin{enumerate}
\item 
Suppose that $U_{f+1}(t) \notin range(trim_f(P(t)))$. Then either $U_{f+1}(t) < \min(trim_f(P(t)))$ or $U_{f+1}(t) > \max(trim_f(P(t)))$.

\begin{itemize}

\item If $U_{f+1}(t) < \min(trim_f(P(t)))$ then there are at least $f+1$ positions ($U_1(t), ..., U_{f+1}(t)$) which are smaller than $\min(trim_f(P(t)))$. This contradicts the definition of $trim_f(P(t))$ (only the $f$ smallest and the $f$ largest elements of P(t) are removed). 

\item If $U_{f+1}(t) > \max(trim_f(P(t)))$ and since $|U(t)| > 2f$ (because $n>3f$), then there are also at least $f+1$ positions in U(t) greater than $\max(trim_f(P(t)))$, which leads to a contradiction.

\end{itemize}

\item The property is symmetric the precedent one.
\end{enumerate}
\end{proof}

\begin{lem}
\label{destination_mid}
Let $D_i(t)$ be the set of destinations computed with 
Algorithm~\ref{alg:1} in systems with $n>3f$. The following properties
hold: (1) $\forall i,~\forall t,~ D_i(t) \leq (U_{f+1}(t) +
U_m(t))/2$ and (2) $\forall i,~\forall t,~ D_i(t) \geq (U_1(t) + U_{m-f}(t))/2.$
\end{lem}

\begin{proof}

Take $d_1$ to be the distance between $U_{f+1}(t)$ and $U_m(t)$.
\begin{enumerate}
\item 
Suppose $D_i(t) > (U_{f+1}(t) + U_m(t))/2$ for some correct robot $i$ at time $t$. 
Then $U_{f+1}(t) < D_i(t)-d_1/2$. And by Lemma~\ref{inside_range}, $U_{f+1}$(t) is inside $range(trim_f(P(t)))$ which means that there is a position inside $range(trim_f(P(t)))$ which is smaller than $D_i(t)-d_1/2$.
Hence there must exists a position inside $range(trim_f(P(t)))$, say p, which is greater than $D_i(t)+d_1/2$ because $D_i$ is the mean of $trim_f(P(t))$. 

$p > D_i(t)+d_1/2$ implies that $p>U_m(t)$, and by lemma~\ref{range-trim} $U_m(t) \geq \max(range(trim_f(P(t))))$ so $p>\max((trim_f(P(t)))$ which contradicts the fact that p is inside $range(trim_f(P(t)))$.

\item Symmetric to the precedent property.
\end{enumerate}

\end{proof}

\begin{lem}
\label{uf-less-st}
Let $S(t)$ be a multiset of $f+1$ arbitrary elements of U(t). We
have the following properties: (1) $\forall t,~ U_{f+1}(t) \leq \max(S(t))$
and (2) $\forall t,~ U_{m-f}(t) \geq \min(S(t))$
\end{lem}

\begin{proof}
\begin{enumerate}
\item Assume to the contrary that $U_{f+1}(t) > \max(S(t))$. This means that $U_{f+1}(t)$ is strictly greater than at least $f+1$ elements of U(t), which leads to a contradiction.
\item The property is symmetric to the precedent.
\end{enumerate}
\end{proof}

\begin{lem}\label{limit-destinations}
Let a time $t_2>t_1$ and let $S(t)$ be a multiset of $f+1$ arbitrary elements in U(t).
If $\forall p \in S(t)$ and  $\forall t \in [t_1, t_2]~ p \leq S_{max}$  
then for each correct robot i in $U(t)$ and for each $t \in [t_1, t_2]~ D_i(t) \leq (S_{max}+U_m(t_1))/2$.
\end{lem}

\begin{proof}
By definition of $S_{max}$ we have that $\forall t \in [t_1, t_2], max(S(t)) \leq S_{max}$.
According to Lemma~\ref{uf-less-st}, $\forall t \in [t_1, t_2]~  U_{f+1}(t) \leq max(S(t))$. 
So $\forall t \in [t_1, t_2]~ U_{f+1}(t) \leq S_{max}$.

By Lemma~\ref{destination_mid}, for each correct robot i and for each $t  \in [t_1, t_2] $, $D_i(t) \leq (U_m(t) + U_{f+1}(t))/2$. 
So for each correct robot i and for each $t  \in [t_1, t_2]$, $D_i(t) \leq (U_m(t)+S_{max})$ .
Since the algorithm is cautious, $ \forall t  \in [t_1, t_2]~ U_m(t) \leq U_m(t_1) $ and the lemma follows.

\end{proof}

\subsection{Convergence of Algorithm \ref{alg:1} in fully-synchronous
ATOM model}

In this section we address the correctness of Algorithm \ref{alg:1} in
the fully-synchronous ATOM model.

\begin{lem}
\label{shrinking-FS}
Algorithm \ref{alg:1} is shrinking for $n>2f$ in fully-synchronous ATOM model.
\end{lem}

\begin{proof}
Let a configuration of robots at time $t$, and let $d_t$ 
be the diameter of correct robots at $t$.
Each cycle, all robots move towards the same destination. They move by
at least a distance of $\delta$ unless they reach their destination.

If all robots are at a distance smaller than $\delta$ from the common
destination point, 
gathering is achieved and the diameter is null. Otherwise, the robots
that are further than $\delta$ from the destination point approach it
by at least $\delta$ so the diameter decreases by at least $\delta$.
Overall, the diameter of robots decreases by at least factor of 
$\alpha=1-(\delta/d_t)$ at each cycle and thus the algorithm is shrinking.
\end{proof}

The correctness of Algorithm \ref{alg:1} follows directly from Lemma \ref{cautious-FS} and Lemma \ref{shrinking-FS}:

\begin{thm}
Algorithm \ref{alg:1} is convergent for $n>2f$ in fully-synchronous ATOM model.
\end{thm}

\subsection{Correctness proof in semi-synchronous ATOM model}
In this section we address the correctness of Algorithm \ref{alg:1} in
semi-synchronous model under a $k$-bounded scheduler. 
Our proof is constructed on top of the 
auxiliary lemmas proposed in the previous sections.

\begin{lem}
\label{shrinking_sym}
Algorithm \ref{alg:1} is shrinking in semi-synchronous
 ATOM model with $n>3f$ under a k-bounded scheduler.
\end{lem}

\begin{proof}

Let $U_1(t_0), ..., U_m(t_0)$ be a configuration of correct robots at
the initial time $t_0$, and assume that they are ordered 
from left to right. Let $d_0$ be the diameter of correct robots at $t_0$,
$d_1=distance(U_{f+1}(t_0), U_m(t_0))$ and $d_2=distance(U_1(t_0)$, 
$U_{m-f}(t_0))$. We assume without loss of generality that $d_1 > d_2$.
Note that in this case $d_1\geq d_0/2$, otherwise $d_1+d_2<d_0$ which
is impossible since $|U(t)| > 2f$.

Let $S(t)$ be the multiset ${U_1(t), ..., U_{f+1}(t)}$. We have at $t_0$: $\max(S(t_0))=U_m(t_0)-d_1$. 
Let $t_1\geq t_0$ be the first time all correct robots
have been activated at least once since $t_0$. 
We prove in the following that at $t_1$, $\max(S(t_1)) \leq U_m(t_0)-d_1/2^{k(f+1)}$.
 
According to Lemma~\ref{uf-less-st}, $\forall t \in [t_0, t_1]~ U_{f+1}(t) \leq
max(S(t))$ and by Lemma~\ref{destination_mid}, $D_i(t) \leq (U_m(t) +
U_{f+1}(t))/2$ for 
each correct robot $i$ and for each $t  \in [t_1, t_2] $. 
So $D_i(t) \leq (U_m(t) + max(S(t)) )/2$.
Since the algorithm is cautious, $\forall t>t_0~ U_m(t) \leq U_m(t_0) $.
So $D_i(t) \leq (U_m(t_0) + max(S(t)) )/2$  for each correct robot $i$
and for each $t  \in [t_1, t_2] $. 

Recall that initially $\max(S(t_0))=U_m(t_0)-d_1$.
Therefore, when at some time $t^\prime > t_0$, a robot in
$S(t^\prime)$ is activated, its calculated 
destination is smaller than $(U_m(t_0) + \max(S(t^\prime)) )/2$. 
Then $\max(S(t^\prime+1)) \leq (U_m(t_0) + \max(S(t^\prime)) )/2$.

Recall that $t_1$ is the first time such that all robots are activated at least once since $t_0$. Since the scheduler is k-bounded, the robots in $S(t)$ may have been activated at most k times each. So between $t_0$ and $t_1$, there are at most $k(f+1)$ activations of robots in $S(t)$. 
Therefore at $t_1$, 
$\max(S(t_1)) \leq (U_m(t_0) - d_1/2^{k(f+1)})$. And since $d_1>d_0/2, \max(S(t_1)) \leq (U_m(t_0) - d_0/2^{k(f+1)+1})$. 

So between $t_0$ and $t_1$ all robots are activated at least once, and
according to Lemma~\ref{limit-destinations}, all their calculated
destinations are less than or equal to $(U_m(t_0) - d_0/2^{k(f+1)+2})$.
Since robots are guaranteed to move toward their destinations by at
least a distance $\delta$ before they can be stopped by the scheduler,
at $t_1$, all the positions of 
$U(t_1)$ are $\leq U_m(t_0)-min\{\delta, d_0/2^{k(f+1)+2}\}$.\\
Thus by setting $\alpha=\max\{1-\delta/d_0,1-1/2^{k(f+1)+2}\}$ at $t_1$, the lemma follows.
\end{proof}

The convergence proof of Algorithm \ref{alg:1} directly follows from 
Lemma \ref{shrinking_sym} and Lemma \ref{cautious-FS}.

\begin{thm}
Algorithm \ref{alg:1} is convergent in semi-synchronous ATOM model for
$n>3f$ under a k-bounded scheduler.
\end{thm}

\section{Concluding remarks}
\label{sec:conclusion}

We studied the problem of convergence of mobile oblivious robots in a uni-dimensional space when some of the robots can exhibit arbitrary malicious behavior. We showed that there is a tradeoff between system synchrony (how tightly synchronized the robots are) and malicious tolerance, as more asynchronous systems lead to less Byzantine tolerance. One originality of our approach is the connection with previous results in fault-tolerant distributed computing with respect to approximate Byzantine agreement. Three immediate open questions are raised by our study:
\begin{enumerate}
\item we consider a uni-dimensional space, which leads to questioning the applicability of our approach in multi-dimensional spaces,
\item we presented lower bound for the class of cautious algorithms, which leaves the possibility of non-cautious solutions for the same problem open,
\item the model we consider in this paper is either fully-synchronous or semi-synchronous, which leads to the possible investigation of purely asynchronous models for the same problem (\emph{e.g.} CORDA~\cite{Pre05}).
\end{enumerate}

\bibliographystyle{plain}
\bibliography{convergence}
 
\end{document}